\documentclass[letterpaper, 12pt]{colt2021} %
\usepackage{mathtools}
\usepackage{enumerate}
\usepackage{float}
\usepackage{bbm}
\usepackage{caption}
\usepackage[utf8]{inputenc}
\usepackage[english]{babel}
\usepackage{dsfont}	
\usepackage{tikz}
\usepackage{dirtytalk}
\usetikzlibrary{arrows,arrows.meta}

\def\m{\mathcal}

\def\Bern{\mathsf{Bern}}

\DeclareMathOperator{\E}{\mathbb{E}}
\DeclareMathOperator{\Pe}{\mathsf{P_e}}

\DeclareMathOperator{\ind}{\mathds{1}}

\DeclareMathOperator{\isit}{\mathsf{RUNS}}

\newcommand{\eps}{\epsilon}

\title[Deterministic Finite-Memory Bias Estimation]{Deterministic Finite-Memory Bias Estimation} 

\usepackage{times}

\coltauthor{%
 \Name{Tomer Berg} \Email{tomerberg@mail.tau.ac.il}\\
 \addr School of Electrical Engineering, Tel Aviv University
 \AND
 \Name{Or Ordentlich} \Email{or.ordentlich@mail.huji.ac.il}\\
 \addr School of Computer Science and Engineering, Hebrew University of Jerusalem
 \AND
 \Name{Ofer Shayevitz} \Email{ofersha@eng.tau.ac.il}\\
 \addr School of Electrical Engineering, Tel Aviv University}
\begin{document}

\maketitle

\begin{abstract}%
 In this paper we consider the problem of estimating a Bernoulli parameter using finite memory. Let $X_1,X_2,\ldots$ be a sequence of independent identically distributed Bernoulli random variables with expectation $\theta$, where $\theta \in [0,1]$. Consider a finite-memory deterministic machine with $S$ states, that updates its state $M_n \in \{1,2,\ldots,S\}$ at each time according to the rule $M_n = f(M_{n-1},X_n)$, where $f$ is a deterministic time-invariant function. Assume that the machine outputs an estimate at each time point according to some fixed mapping from the state space to the unit interval. The quality of the estimation procedure is measured by the asymptotic risk, which is the long-term average of the instantaneous quadratic risk. The main contribution of this paper is an upper bound on the smallest worst-case asymptotic risk any such machine can attain. This bound coincides with a lower bound derived by Leighton and Rivest, to imply that $\Theta(1/S)$ is the minimax asymptotic risk for deterministic $S$-state machines. In particular, our result disproves a longstanding $\Theta(\log S/S)$ conjecture for this quantity, also posed by Leighton and Rivest.
\end{abstract}


\section{Introduction}

The statistical hardness of a parametric estimation problem has been traditionally characterized by the number of independent samples from the distribution $P_{\theta}$ one needs to see in order to accurately estimate $\theta$. However, as the amount of available data is constantly increasing, collecting enough samples for accurate estimation is becoming less of a problem, provided that the parameter $\theta$ is of a relatively low dimension. In this regime, it is the computational resources dedicated to the estimation task, rather than the number of samples, that constitute the main bottleneck determining the quality of estimation one can attain.

As a result, the topic of estimation / learning under computational constraints is currently drawing considerable attention; in particular, the problem of estimation / learning under memory constraints has been recently studied in various different setups, as we further elaborate in Subsection~\ref{subsec:related_work}. Despite this ongoing effort, there are still substantial gaps in the understanding of the effects memory limitations can have on the minimal possible estimation error. This work addresses such a gap in arguably the simplest setup possible: estimation of a single parameter $\theta\in[0,1]$ from an infinite number of independent samples from $P_\theta$, using a finite-memory learning algorithm.



Specifically, we consider the bias estimation problem, defined as follows: $X_1,X_2,\ldots$ is a sequence of independent identically distributed random variables drawn according to the $\Bern(\theta)$ distribution. An \textit{$S$-state estimation procedure} for this problem  consists of two functions: $f$, and $\hat{\theta}$, where $f:[S] \times \{0,1\}\rightarrow [S]$ is a deterministic state transition (or memory update) function (here $[S]=\{1,\ldots,S\}$), and
$\hat\theta:[S]\rightarrow [0,1]$ is the estimate function. Letting $M_n$ denote the state of the memory at time $n$, this finite-state machine evolves according to the rule 
\begin{align}
M_0&=s_{\text{init}},\label{eq:init} \\M_n&=f(M_{n-1},X_n)\in [S],\label{eq:evolution}
\end{align}
for some predetermined initial state $s_{\text{init}} \in [S]$. If the machine is stopped at time $n$, it outputs the estimation $\hat\theta(M_n)$. We define the (asymptotic) quadratic risk attained by this estimation procedure, given that the true value of the parameter is $\theta$, to be\footnote{It is not difficult to show that the limit exists due to the underlying finite-state structure and the independence of the samples.}
\begin{align}
R_\theta(f,\hat\theta) = \lim_{n\rightarrow\infty} \frac{1}{n}\sum_{i=1}^n \E \left(\hat{\theta}(M_i)-\theta\right)^2.\label{eq:MSE} 
\end{align}
We are interested in the \emph{minimax risk} of the estimation problem, defined as 
\begin{align}
R(S) \triangleq \min_{ f,\hat\theta}\max_{\theta\in[0,1]} R_\theta (f,\hat\theta), 
\end{align} 
where the minimum is taken over all $S$-state estimation procedures. This paper derives an upper bound on the minimax risk, which together with a known lower bound, establishes the exact behavior of the minimax risk with $S$. 
%

%

Note that here the memory update function $f$ is not allowed to depend on time. First, as our focus here is on upper bounds, it is always desirable to use the weakest possible model. Moreover, the restriction to time-invariant algorithms is operationally appealing, since storing the time index necessarily incurs a memory cost. Furthermore, since the number of samples is unbounded, just storing the code generating a time-varying algorithm may require unbounded memory. 

Besides memory, another resource that plays an important role here is randomness. While allowing the use of randomization in the estimation procedure may certainly help, this resource has a cost. Even if one has access to unlimited randomness (which is the case in our setting, since randomness can be easily extracted from the i.i.d. sequence $X_1,X_2,\ldots$), storing this randomness places a toll on one's memory budget, which needs to be taken into account in our deterministic setup. One can nevertheless define the \textit{randomized minimax risk} of the estimation problem, to be the smallest asymptotic risk that can be uniformly guaranteed when randomized state-transition functions $f$ are allowed, i.e., 
\begin{align}
R_{\mathsf{rand}}(S) \triangleq \min_{ \text{randomized } f,\hat\theta}\,\max_{\theta\in[0,1]} R_\theta (f,\hat\theta), 
\end{align} 
We emphasize that in the above, in contrast to the deterministic setup we consider in this paper, randomness is ``free'' and not counted toward the memory budget. Our main result is that, in contrast to what was conjectured by~\cite{leighton1986estimating}, $R(S)$ and $R_{\mathsf{rand}}(S)$ are equal up to constants independent of $S$.

Let us be more precise. Prior to this work, it was known that  $R_{\mathsf{rand}}(S)=\Theta(1/S)$. The upper bound was proved by~\cite{samaniego1973estimating}, who constructed an $S$-state randomized estimation procedure that asymptotically induces a $\mathrm{Binomial}(S-1,\theta)$ stationary distribution on the memory state space. The lower bound was established a decade later by~\cite{leighton1986estimating}, using the Markov chain tree theorem. In the same paper, ~\cite{leighton1986estimating} further constructed a deterministic $S$-state estimation procedure by de-randomizing Samaniego's construction, and as a result showed that $R(S) = O(\log{S}/S)$. 
%
They then conjectured that this is the best possible asymptotic minimax risk any deterministic $S$-state estimation procedure can attain, and further stated the problem of proving or disproving this conjecture as the first out of five open problems left for future research. A nice interpretation of their conjecture is the naturally appealing claim that an optimal deterministic algorithm can be obtained by de-randomizing the optimal random algorithm. In their deterministic algorithm, which they believed to be optimal, randomness is extracted from the measurements by augmenting each state with $O(\log(S))$ additional states, which increases the overall MSE (see Section III of~\cite{leighton1986estimating}). Surprisingly, we show that such a de-randomization is suboptimal,
thereby disproving the conjecture of Leighton and Rivest.




\begin{theorem}\label{thm:good_det}
	\begin{align}
R(S) = O\left(\frac{1}{S}\right).
	\end{align}
\end{theorem}
Since deterministic $S$-state estimation procedures are a subset of the class of $S$-state randomized estimation procedures, we clearly have $R (S)\geq R_{\mathsf{rand}}(S)=\Omega(1/S)$, where the lower bound is due to~\cite{leighton1986estimating}. Consequently:
\begin{corollary}
	\begin{align}
	R(S) = \Theta\left(\frac{1}{S}\right).   
	\end{align}
\end{corollary}

\subsection{Related work}\label{subsec:related_work}

The study of learning and estimation under memory constraints has been initiated in the late 1960s by Cover and Hellman (with a precursor by~\cite{robbins1956sequential}) and remained an active research area for a decade or so. It has then been largely abandoned, but recently it has been again enjoying much attention, due to the reasons described above, and many works have addressed different aspects of the learning under memory constraints problem over the last few years. See, e.g.,~\cite{sd15,svw16,raz18,ds18,dks19,ssv19} for a far from exhaustive list of recent works. 

Most of the old work on learning with finite memory has been focused on the hypothesis testing problem. For the problem of deciding whether an i.i.d. sequence was drawn from $\Bern(p)$ or $\Bern(q)$,~\cite{cover1969hypothesis} described a time-varying finite-state machine with only $S=4$ states, whose error probability approaches zero with the sequence length. As time-varying procedures suffer from the shortcomings described above,~\cite{hellman1970learning} addressed the same binary hypothesis testing problem within the class of \textit{time-invariant randomized} procedures. They have found an \emph{exact} characterization of the smallest attainable error probability for this problem. To demonstrate the important role randomization plays in approaching this value, the same authors show in~\cite{hellman1971memory} that for any memory size $S<\infty$ and $\delta>0$, there exists problems such that any $S$-state deterministic procedure has probability of error $\Pe\geq \frac{1}{2}-\delta$, while their randomized procedure from~\cite{hellman1970learning} has $\Pe\leq \delta$. Note that one can simulate a randomized procedure with a deterministic one by using some of the samples of $\{X_n\}$ for randomness extraction, e.g., using~\cite{von195113} extraction. However, the extracted random bits must be stored, which could result in a substantial increase in memory, see~\cite{chandrasekaran1970finite}. In a recent paper~\cite{berg2020binary} derived a lower bound on the error probability attained by any $S$-state deterministic procedure, showing that while the smallest attainable error probability decreases exponentially fast with $S$ in both the randomized and the deterministic setups, the base of the exponent can be arbitrarily larger in the deterministic case.

One of the earlier works on estimation with finite memory is due to~\cite{roberts1970estimation}, who tackled the problem of estimation under quadratic risk for a random variable with additive noise. ~\cite{hellman1974finite} studied the problem of estimating the mean $\theta$
of a Gaussian distribution and discovered a $S$-state estimation procedure that asymptotically achieves the same Bayesian  quadratic risk as the optimal $S$-level quantizer $Q(\theta)$ for $\theta$, where $Q:\mathbb{R}\to[S]$.
As already described above,~\cite{samaniego1973estimating} and~\cite{leighton1986estimating} have showed that $R_{\mathsf{rand}}(S)=\Theta(1/S)$.~\cite{meron2004finite,ingber2006prediction,dar2014finite} studied the subject of finite-memory universal prediction of sequences using randomized/deterministic machines. More recently,~\cite{jain2018effective} studied the shrinkage in memory between the hypothesis testing and the estimation problem, namely the interesting fact that a machine with $S$ states can distinguish between two coins with biases that differ by $1/S$, whereas the best additive accuracy it can achieve in estimating the bias is only $1/\sqrt{S}$. We further note that the problem of estimating statistics with bounded memory is attracting considerable attention in the machine learning literature lately, see, e.g.,~\cite{chien2010space,kontorovich2012statistical,mcgregor2012space,sd15,svw16,raz18,ds18,dks19,ssv19}. Another closely related active line of work is that of estimating statistics under limited communication, e.g.,~\cite{zdjw13,garg2014communication,bgmnw16,xr17,jordan2018communication,han2018geometric,how18,bho18,act18,hadar2019communication,hadar2019distributed,acharya2020inference}.

\section{Proof of Theorem~\ref{thm:good_det}}

We now proceed to prove Theorem~\ref{thm:good_det}. We will describe our deterministic $S$-state estimation procedure and show that it attains quadratic risk of $O(1/S)$ uniformly for all $\theta\in[0,1]$. In this section we provide the entire proof, but for clarity we rely on several technical claims whose proofs are relegated to the next section or to the Appendix.


Recall from~\eqref{eq:init} and~\eqref{eq:evolution} that any deterministic $S$-state estimation procedure corresponds to a finite-state machine with $S$ states, with at most two outgoing edges from each state, one for $X_i=0$ and one for $X_i=1$. Running this machine on an i.i.d. $\Bern{(\theta)}$ input sequence $X_1,X_2,\ldots$, generates a Markov chain $\{M_n\}_{n=1}^\infty$, where $M_n$ denotes the state of the machine at time $n$. We emphasize that the distribution of the process $\{M_n\}$ depends on $\theta$, which is the parameter we are trying to estimate. To lighten notation, we nevertheless leave this dependence implicit. The construction we describe below trivially achieves  $R_\theta(f,\hat{\theta})=O(1/S)$ for $\theta=0$ and $\theta=1$, and thus in  the remainder of the paper we assume without loss of generality that $\theta\in(0,1)$.

The high-level idea underlying our scheme is to break down the memory-constrained  estimation task into a sequence of memory-constrained (composite) binary hypothesis testing sub-problems. In each such sub-problem, the goal is to decide whether the true underlying parameter $\theta$ satisfies $\{\theta<q\}$ or $\{\theta>p\}$, for some $0<q<p<1$. Those decisions are then used in order to traverse an induced Markov chain in a way that enables us to accurately estimate $\theta$. 


Let us now describe the particular structure of the proposed machine. In our construction, the state space $[S]$ is partitioned into $K$ disjoint sets denoted by $\mathcal{S}_1,\ldots,\mathcal{S}_K$, where the estimation function value is the same inside each $\mathcal{S}_k$, i.e., 
\begin{align}
\hat{\theta}(s)=\hat{\theta}_k,~~~\forall s\in\mathcal{S}_k, \;k\in[K].
\end{align}
The goal is to design a machine for which the stationary distribution of $\{M_n\}$ corresponding to the parameter $\theta$ will concentrate on states that belong to classes $\mathcal{S}_k$ for which $(\theta-\hat{\theta}_k)^2$ is the smallest. This goal is in general easier to achieve when each set consists of a large number of states, which corresponds to small $K$ (as the total number of states $S$ is fixed). On the other hand, the quadratic risk such a machine can attain is obviously limited by the number of different sets $K$, and in particular is $\Omega(1/K^2)$, as there must exist some $\theta\in[0,1]$ at distance $\Omega(1/K)$ from all points $\hat{\theta}_1,\ldots,\hat\theta_K$. Thus, the choice of $K$ should balance the tension between these two contrasting goals; specifically, we will see that the choice $K=\Theta(\sqrt{S})$ is suitable to that end.

Since the estimator $\hat{\theta}$ depends on $\{M_n\}$ only through its class, it is natural to define the \emph{quantized process} $\{Q_n\}_{n=1}^\infty$ obtained by the deterministic scalar mapping \begin{align}
Q_n=\phi(M_n),~~~n=1,2,\ldots,    \label{eq:QnDef}
\end{align}
where $\phi:[S]\to[K]$ maps each state to its set label (namely: $\phi(s)=k$ iff $s\in\mathcal{S}_k$).
The process $\{Q_n\}$, as well as any process on a finite alphabet, consists of \emph{runs} of the same letter. We can therefore decompose it as $\{S_1,\tau_1\},\{S_2,\tau_2\},\ldots$, where $S_i$ denotes the first letter in the $i$th run, and $\tau_i$ denotes its length. We refer to the process $\{S_i\}_{i=1}^\infty$, supported on $[K]$ as the \emph{sampled process}, and to $\{\tau_i\}_{i=1}^\infty$, supported on $\mathbb{N}$, as the \emph{holding times} process. Note that both $\{S_i\}$ and $\{\tau_i\}$ are deterministically determined by $\{Q_n\}$ and hence, by the original process $\{M_n\}$. 
In general, the sampled process can be complicated; however, in our construction, we impose a particular structure that ensures that the sampled process $\{S_n\}$ is also a Markov process. Specifically, for each $k\in[K]$ there is an \textit{entry state} $s_k\in\mathcal{S}_k$, such that all edges going out of a state $\ell\notin\mathcal{S}_k$ to the set $\mathcal{S}_k$, go into the entry state $s_k\in\mathcal{S}_k$. In other words, whenever $M_n$ enters the set $\mathcal{S}_k$ from a different set, it does so through the designated entry state only. This feature guarantees that at the beginning of the $i$th run, the state of the original process $\{M_n\}$ is determined by $S_i$, and consequently $\{S_i\}$ is indeed a Markov process itself. Furthermore, conditioned on $S_i$, the holding time $\tau_i$ is independent of the entire past. We denote the conditional distribution of $\tau_i$ conditioned on the event $S_i=k$, by $P_{T_k}$. It will be convenient to also define the random variables $T_k\sim P_{T_k}$, for $k\in[K]$. In our construction, we further guarantee that any set $\mathcal{S}_k$ is accessible from any other set $\mathcal{S}_j$, $j\neq k$. This ensures that the underlying Markov process $\{M_n\}$ is ergodic, and as a result, so is the sampled process $\{S_n\}$. We refer to the structure described here, i.e., all sets are accessible from one another and have entry states, as a \emph{nested Markov structure}. 

The ergodicity of $\{M_n\}$ immediately implies the ergodicity of the quantized process $\{Q_n\}$, by~\eqref{eq:QnDef}. Denote by $\pi_k$ the stationary probability of state $k$ for the process $\{Q_n\}$. We therefore have that for a machine $f,\hat{\theta}$ of the type described above,
\begin{align}
	    R_\theta = R_\theta(f,\hat\theta) = \lim_{n\rightarrow\infty} \frac{1}{n}\sum_{i=1}^n \E \left(\hat{\theta}(M_i)-\theta\right)^2=\sum_{k=1}^K \pi_k \left(\hat{\theta}_k-\theta\right)^2.\label{eq:MSEgen}
	\end{align}
The next lemma determines the stationary distribution $\{\pi_k\}_{k\in[K]}$ of the quantized process $\{Q_n\}$, in terms of the stationary distribution $\{\mu_k\}_{k\in[K]}$ of the sampled process $\{S_n\}$ and the expected holding times $\{\E[T_k]\}_{k\in[K]}$.
\begin{lemma}\label{lem:Qn_stat}
The unique stationary probability of state $k$ under the process $\{Q_n\}$ is
\begin{align}
    \pi _k =\frac{\E[T_k] \mu_k }{\sum_{j=1}^M\E[T_j] \mu _j}.
\end{align}
\end{lemma}
Combining Lemma~\ref{lem:Qn_stat} with~\eqref{eq:MSEgen}, we have that the risk of such machine is
\begin{align}
     R_\theta &=\sum_{k=1}^K\frac{\E[T_k] \mu_k }{\sum_{j=1}^M\E[T_j] \mu _j}\left(\hat{\theta}_k-\theta\right)^2.\label{eq:MSEnestMark}
\end{align}
It is immediately evident from~\eqref{eq:MSEnestMark} that the asymptotic risk attained by a  machine with the nested Markov structure defined above depends only on the stationary distribution of the sampled process $\{S_n\}$ and the expected holding times. Ideally, we would like to construct this machine such that two things would happen for every $\theta$: 
\begin{enumerate}[(i)]
    \item $\{\mu_k\}$ would be concentrated on states whose corresponding estimate $\hat{\theta}_k$ is close to $\theta$;\label{prop1}
    \item The expected holding times for these states will be at least as large as those of other states. \label{prop2}
\end{enumerate}
 
We now describe how our machine is designed to achieve the desired behaviour~\eqref{prop1} of $\{\mu_k\}$. Later, we will see that the desired behavior~\eqref{prop2} of $\{\mathbb{E}[T_k]\}$ follows more or less automatically. 

First, we set our estimators to be\footnote{The denominator is set to $K+2$ rather than $K$ for minor technical reasons, in order to avoid dealing with probabilities on the boundary of the simplex in the analysis.} 
\begin{align}
\hat{\theta}_k=\frac{k}{K+2},~~~k\in[K].
\end{align}
We then design our machine such that the sampled process $\{S_n\}$ is a random walk, that moves either one state left or one state right from each state (except for the extreme states $1$ and $K$ that behave slightly differently). In particular, the $k$th state in $\{S_n\}$ is connected only to states $k+1$ and $k-1$ for all $k\in\{2,\ldots,K-1\}$. The precise diagram for the sampled process $\{S_n\}$ is shown in Figure~\ref{fig:sample_chain}, where the transition probabilities $\{p_k,q_k=1-p_k\}_{k\in[K]}$ will depend on $\theta$ through the construction of the original machine generating the original Markov chain $\{M_n\}$. We design the machine in a way that guarantees that the random walk $\{S_n\}$ has a strong \emph{drift} towards the state $k$ whose corresponding estimator is closest to $\theta$. In particular, if $\theta>\frac{k+1}{K+2}$ then $p_k>1-\eps$  and conversely, if $\theta<\frac{k}{K+2}$ then $p_k<\eps$, for some $\eps<1/2$ and all states $k\in\{2,\ldots,K-1\}$.

\begin{center}
\begin{figure}[H]
\centering
\setlength\belowcaptionskip{-1.4\baselineskip}
\begin{tikzpicture}
  \tikzset{
    >=stealth',
    node distance=0.92cm,
    state/.style={font=\scriptsize,circle, align=center,draw,minimum size=20pt},
    dots/.style={state,draw=none}, edge/.style={->},
  }
  \node [state] (S0) {$1$};
  \node [state] (S0-1)   [right of = S0]   {};
  \node [dots] (dots1)   [right of = S0-1]   {$\cdots$};
  \node [state] (1l) [right of = dots1]  {};
  \node [state ,label=center:$i$] (0)   [right of = 1l] {};
  \node [state] (1r) [right of = 0]   {};
  \node [dots]  (dots2)  [right of = 1r] {$\cdots$};
  \node [state] (S1-1) [right of = dots2]  {};
  \node [state] (S1) [right of = S1-1] {$K$};
  \path [->,draw,thin,font=\footnotesize]  (S0-1) edge[bend left=45] node[below ] {$q_2$} (S0);
  \path [->,draw,thin,font=\footnotesize]  (0) edge[bend left=45] node[below] {$q_i$} (1l);
  \path [->,draw,thin,font=\footnotesize]  (S1) edge[bend left=45] node[below right] {$1$} (S1-1);
  \path [->,draw,thin,font=\footnotesize]  (1r) edge[bend left=45] node[below right] {$q_{i+1}$} (0);
  \path [->,draw,thin,font=\footnotesize]  (S1-1) edge[bend left=45] node[below right] {$q_{K-1}$} (dots2);
  \path [->,draw,thin,font=\footnotesize]  (1l) edge[bend left=45] node[below ] {$q_{i-1}$} (dots1);
  \path [->,draw,thin,font=\footnotesize]  (S0-1) edge[bend left=45] node[above left] {$p_2$} (dots1);
  \path [->,draw,thin,font=\footnotesize]  (1l) edge[bend left=45] node[above left] {$p_{i-1}$} (0);
  \path [->,draw,thin,font=\footnotesize]  (0) edge[bend left=45] node[above ] {$p_i$} (1r);
  \path [->,draw,thin,font=\footnotesize]  (S0) edge[bend left=45] node[above ] {$1$} (S0-1);
  \path [->,draw,thin,font=\footnotesize]  (1r) edge[bend left=45] node[above] {$p_{i+1}$} (dots2);
  \path [->,draw,thin,font=\footnotesize]  (S1-1) edge[bend left=45] node[above ] {$p_{K-1}$} (S1);
\end{tikzpicture}
\caption{A sampled chain of $K$ states.}\label{fig:sample_chain}
\end{figure}
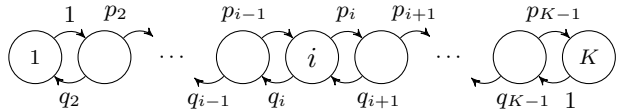
\end{center}

The desired drift behavior is enabled by constructing the sets $\mathcal{S}_1,\ldots,\mathcal{S}_K$ as \emph{mini-chains}, where the $k$th mini-chain consists of $N_k$ states, and is designed to solve the composite binary hypothesis testing problem: $\mathcal{H}_0:\left\{\theta>\frac{k+1}{K+2}\right\}$ vs. $\mathcal{H}_1:\left\{\theta<\frac{k}{K+2}\right\}$. Each mini-chain $\mathcal{S}_k$ is initialized in its entry state $s_k$, and eventually moves to the entry state $s_{k+1}$ of mini-chain $\mathcal{S}_{k+1}$ if it decided in favor of hypothesis $\mathcal{H}_0$, or to the entry state $s_{k-1}$ of mini-chain $\mathcal{S}_{k-1}$ if it decided in favor of hypothesis $\mathcal{H}_1$. The time it takes it to ``make a decision'' is the random holding time with some distribution $P_{T_k}$. Note that if the error probability of the machine is smaller than $\eps<1/2$ under both hypotheses, we will indeed attain the desired drift behavior. Our goal now is to design mini-chains that attain small error probabilities with as few states as possible. To that end, we appeal to~\cite{berg2020binary}, where the authors defined the following machine.\footnote{Their machine was designed to solve the \emph{simple} binary hypothesis test $\mathcal{H}_0:\{\theta=p\}$ vs. $\mathcal{H}_1:\{\theta=q\}$, but as our analysis demonstrates, the difference between the two problems is not significant.}


\begin{definition}
  $\isit(N,p,q)$ is the machine with $N\geq 4$ states depicted in Figure~\ref{fig:BHT_Machine},
  designed to decide between the hypotheses $\mathcal{H}_0:\{\theta>p\}$ vs. $\mathcal{H}_1:\{\theta<q\}$, for some $0<q<p<1$. The machine is initialized at state $s$ and evolves according to the sequence of input bits $X_1,X_2,\ldots$. If the machine observes a run of $N-s$ ones before observing a run of $s-1$ zeros, it decides $\mathcal{H}_0$ and exists right. Otherwise, it decides $\mathcal{H}_1$ and exists left.
  The initial state of the machine is $s=f(N,p,q)$, where
  \begin{align}
    f(N,p,q) \triangleq  2+\left\lceil \frac{\log pq}{\log p(1-p)+\log q(1-q)}(N-3)\right\rfloor,\label{eq:start_state}
  \end{align}
  is an integer between $2$ and $N-1$. We denote the (worst case) error probability of the machine by $\Pe^{\isit(N,p,q)}=\max\left\{p^0_1,p^1_0\right\}$, where
  \begin{align}
  p^0_1&=\sup_{\theta<q}~~\Pr_{X_1,X_2\ldots\stackrel{i.i.d.}{\sim}\Bern(\theta)}\left(\isit(N,p,q)\text{ decides } \mathcal{H}_0 \right),\\
  p^1_0&=\sup_{\theta>p}~~\Pr_{X_1,X_2\ldots\stackrel{i.i.d.}{\sim}\Bern(\theta)}\left(\isit(N,p,q)\text{ decides } \mathcal{H}_1 \right).
  \end{align}


\end{definition}
\begin{center}
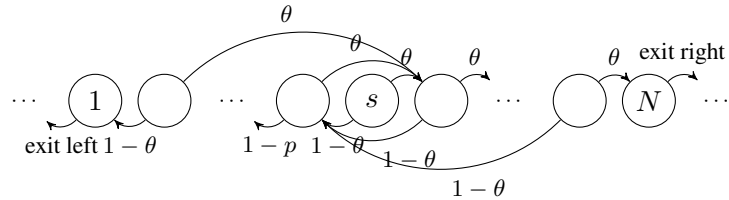
\begin{figure}[H]
\centering
\setlength\belowcaptionskip{-1.4\baselineskip}
\begin{tikzpicture}
  \tikzset{
    >=stealth',
    node distance=0.92cm,
    state/.style={font=\scriptsize,circle, align=center,draw,minimum size=20pt},
    dots/.style={state,draw=none}, edge/.style={->},
  }
  \node [state ,label=center:$1$] (S0)  {} ;
  \node [state] (S0-1)   [right of = S0]   {};
  \node [dots] (dots1)   [right of = S0-1]   {$\cdots$};
  \node [state] (1l) [right of = dots1]  {};
  \node [state ,label=center:$s$] (0)   [right of = 1l] {};
  \node [state] (1r) [right of = 0]   {};
  \node [dots]  (dots2)  [right of = 1r] {$\cdots$};
  \node [state] (S1-1) [right of = dots2]  {};
  \node [state ,label=center:$N$] (S1) [right of = S1-1]  {};
   \node [dots]  (dots3)  [right of = S1] {$\cdots$};
    \node [dots]  (dots4)  [left of = S0] {$\cdots$};
  \path [->,draw,thin,font=\footnotesize]  (S0-1) edge[bend left=45] node[below ] {$1-\theta$} (S0);
  \path [->,draw,thin,font=\footnotesize]  (0) edge[bend left=45] node[below] {$1-\theta$} (1l);
  \path [->,draw,thin,font=\footnotesize]  (1r) edge[bend left=45] node[below right] {$1-\theta$} (1l);
  \path [->,draw,thin,font=\footnotesize]  (S1-1) edge[bend left=45] node[below right] {$1-\theta$} (1l);
  \path [->,draw,thin,font=\footnotesize]  (1l) edge[bend left=45] node[below ] {$1-p$} (dots1);
  \path [->,draw,thin,font=\footnotesize]  (S0-1) edge[bend left=45] node[above left] {$\theta$} (1r);
  \path [->,draw,thin,font=\footnotesize]  (1l) edge[bend left=45] node[above left] {$\theta$} (1r);
  \path [->,draw,thin,font=\footnotesize]  (0) edge[bend left=45] node[above ] {$\theta$} (1r);
  \path [->,draw,thin,font=\footnotesize]  (1r) edge[bend left=45] node[above] {$\theta$} (dots2);
   \path [->,draw,thin,font=\footnotesize]  (S0) edge[bend left=45] node[below] {exit left} (dots4);
  \path [->,draw,thin,font=\footnotesize]  (S1-1) edge[bend left=45] node[above ] {$\theta$} (S1);
 \path [->,draw,thin,font=\footnotesize]  (S1)  edge[bend left=45]  node[above ]{exit right} (dots3);

\end{tikzpicture}
\caption{$\isit(N,p,q)$ - Deterministic Binary Hypothesis Testing Machine  }\label{fig:BHT_Machine}
\end{figure}
\end{center}


The next lemma demonstrates that with $N=O(K)$ states, the machine $\isit(N,p,q)$ can decide whether $\theta>p$ or $\theta<q=p-1/K$ with constant error probability $\eps<1/2$.
Thus, the desired drift can be attained by mini-chains of $O(K)$ states.

\begin{lemma} \label{lem:N_eps}
For any $\frac{2}{K}\leq  p\leq 1-\frac{1}{K}$, $q=p-\frac{1}{K}$ and $0<\eps<1/2$, let\footnote{Logarithms in this paper are taken to base $2$.} 
\begin{align}
N = N(\eps,p,K) \triangleq 3+\left\lceil K\cdot 6\log\frac{2}{\eps\cdot\left(p-\frac{1}{K}\right)(1-p)} \right\rceil.
\label{eq:NepK}
\end{align}
Then
\begin{align}
\Pe^{\isit(N,p,q)}<\eps.
\end{align}
\end{lemma}

We therefore take the $k$th mini-chain $\m{S}_k$ as the machine $\isit(N_k,p,q)$ with $q=\frac{k}{K+2}$,  $p=q+\frac{1}{K+2}=\frac{k+1}{K+2}$, and $N_k=N(\eps,\frac{k+1}{K+2},K+2)$. The total number of states in our machine is therefore (see calculation in the appendix)
\begin{align}
S&=\sum_{k=1}^K N_k=\sum_{k=1}^K N\left(\eps,\frac{k+1}{K+2},K+2\right)\leq 6(K+2)^2\log \left(\frac{2e}{\epsilon}\right)\label{eq:Ssum},
\end{align}
and the sampled chain $\{S_n\}$ indeed satisfies the desired drift property: for all $2\leq k\leq K-1$ we have that if $\theta>\frac{k+1}{K+2}$ then $p_k>1-\eps$ whereas if $\theta<\frac{k}{K+2}$ then $p_k<\eps$. Note that we did not quantify $p_k$ for the case where $\theta\in\left[\frac{k}{K+2},\frac{k+1}{K+2}\right]$, but as will become apparent below, it is indeed not needed for our analysis. Also note that whenever the sampled chain reaches state $1$ it will immediately move back to state $2$, and whenever it reaches state $K$ it will immediately move back to state $K-1$ (that is, $p_1=1$ and $p_K=0$), but the holding times in those states are nevertheless random (and may be very large if the underlying $\theta$ is very close to $0$ or $1$, and dictated by the time it takes for the corresponding $\isit(N,p,q)$ mini-chains $\m{S}_1$ and $\m{S}_K$ to reach a decision). The next lemma shows that the drift property implies that if $\theta\in\left[\frac{k}{K+2},\frac{k+1}{K+2}\right]$, then the stationary probability $\mu_j$ of the $j$th state in the sampled chain decreases exponentially with the ``distance'' $|j-k|$.



\begin{lemma}\label{lem:trans_prob}
Assume that $\theta\in\left[\frac{k}{K+2},\frac{k+1}{K+2}\right]$. Then, the stationary distribution of the sampled process $\{S_n\}$ induced by the machine described above satisfies
\begin{align}
  \mu_{k-i}\leq \mu_{k-1}\left(\frac{\epsilon}{1-\epsilon}\right)^{i-1}  
\end{align}
for $1 \leq i \leq k-1$, and
\begin{align}
  \mu_{k+i}\leq  \mu_{k+1}\left(\frac{\epsilon}{1-\epsilon}\right)^{i-1}  
\end{align}
for $1\leq i \leq K-k$.
\end{lemma}
This shows that the stationary distribution of the sampled chain $\{S_n\}$ is indeed concentrated on the desired states. The next lemma deals with the expected holding times, and lower bounds the ratio between the expected holding time in the ``correct state'' $k$ and the expected holding time in any other state of the sampled chain. 
\begin{lemma}\label{lem:wait_time}
If $\theta < \frac{j}{K+2}$,
then the expected holding time in state $i$ satisfies
\begin{align}
    \E [T_j] \geq (1-\epsilon)\E [T_i] \label{eq:uppbound1}
\end{align}
for all $i>j$. Similarly, if $\theta >\frac{j+1}{K+2}$, then the expected holding time in state $i$ satisfies
\begin{align}
    \E [T_j] \geq (1-\epsilon)\E [T_i] \label{eq:uppbound2}
\end{align}
for all $i<j$.
\end{lemma}
We now combine~\eqref{eq:MSEnestMark} with Lemma~\ref{lem:trans_prob} and Lemma~\ref{lem:wait_time} in order to upper bound the asymptotic risk attained by our machine, and establish the claim $R_{\theta}=O(1/S)$ for all $\theta\in(\frac{1}{K+2},\frac{K+1}{K+2})$. The cases where $\theta\in[0,\frac{1}{K+2})$ and $\theta\in(\frac{K+1}{K+2},1]$ then follow easily from minor adjustments, and are treated in the appendix.

Assume that $\frac{k}{K+2}\leq \theta \leq \frac{k+1}{K+2}$ for some $k\in[K]$.  From~\eqref{eq:MSE}, the asymptotic risk is then
\begin{align}
   R_\theta &=\sum_{i=1}^K\frac{\E[T_i] \mu_i }{\sum_{j=1}^K\E[T_j] \mu _j}\left(\frac{i}{K+2}-\theta\right)^2
  \\&=\sum_{i=1}^{k-1}\frac{\E[T_i] \mu_i }{\sum_{j=1}^K\E[T_j] \mu _j}\left(\frac{i}{K+2}-\theta\right)^2+\frac{\E[T_k] \mu_k }{\sum_{j=1}^K\E[T_j] \mu _j}\left(\frac{k}{K+2}-\theta\right)^2\nonumber\\&+\sum_{i=k+1}^K\frac{\E[T_i] \mu_i }{\sum_{j=1}^K\E[T_j] \mu _j}\left(\frac{i}{K+2}-\theta\right)^2
  \\ &\leq \frac{1}{1-\epsilon}\sum_{i=1}^{k-1}\frac{\E[T_{k-1}] \mu_{k-1} }{\sum_{j=1}^K\E[T_j] \mu _j} \frac{\mu_i}{\mu_{k-1}}\left(\frac{i}{K+2}-\theta \right)^2+\frac{1}{(K+2)^2}\nonumber\\&+\frac{1}{1-\epsilon}\sum_{i=k+1}^K\frac{\E[T_{k+1}] \mu_{k+1} }{\sum_{j=1}^K\E[T_j] \mu _j}\frac{\mu_i}{\mu_{k+1}}\left(\frac{i}{K+2}-\theta\right)^2\label{eq:exp_bound}\\&\leq  
 \frac{1}{1-\epsilon}\sum_{i=1}^{k-1} \left(\frac{\epsilon}{1-\epsilon}\right)^{i-1} \left(\frac{i+1}{K+2}\right)^2+\frac{1}{(K+2)^2}+ \frac{1}{1-\epsilon}\sum_{i=1}^{K-k}\left(\frac{\epsilon}{1-\epsilon}\right)^{i-1} \left(\frac{i+1}{K+2}\right)^2\label{eq:stationary_bound}\\&\leq \frac{1}{(K+2)^2}\cdot \frac{1}{1-\epsilon}\left(2\cdot\sum_{i=1}^{\infty} \left(\frac{\epsilon}{1-\epsilon}\right)^{i-1}(i+1)^2+1\right)\label{eq:sum_bound}\\&\leq \frac{6\log \left(\frac{2e}{\epsilon}\right)}{S}\left(\frac{2\epsilon}{(1-2\epsilon)^3}+\frac{8(1-\epsilon)}{(1-2\epsilon)^2}+\frac{1}{1-\epsilon}\right),\label{eq:sum_identity}
\end{align}
where~\eqref{eq:exp_bound} follows from Lemma~\ref{lem:wait_time}, ~\eqref{eq:stationary_bound} follows from Lemma~\ref{lem:trans_prob} and since $\frac{\E[T_j] \mu_j }{\sum_{k=1}^M\E[T_k] \mu _k}\leq 1$,~\eqref{eq:sum_bound} is since we only add positive terms, and~\eqref{eq:sum_identity} is due to the identity $\sum_{i=0}^\infty q^i (i+2)^2=\frac{q(1+q)+4(1-q)}{(1-q)^3}$ and by substituting~\eqref{eq:Ssum}. Finally, substituting $\eps=1/100$ into~\eqref{eq:sum_identity} gives $R_{\theta}<\frac{600}{S}$.

\section{Proofs of Technical Claims}

The following simple lemma will be useful for the proofs of Lemma~\ref{lem:Qn_stat}
and Lemma~\ref{lem:trans_prob}.
\begin{lemma}\label{lem:flow}
Let $\{X_n\}$ be a stationary process over some alphabet $\mathcal{S}$. Then for any disjoint partition $\mathcal{C}\cup \mathcal{C}'=\mathcal{S}$, it holds that
\begin{align}
\Pr(X_n\in\mathcal{C},X_{n+1}\in\mathcal{C}')=\Pr(X_n\in\mathcal{C}',X_{n+1}\in\mathcal{C}).
\end{align}
\end{lemma}
\begin{proof}
For any disjoint partition $\mathcal{C}\cup \mathcal{C}'=\mathcal{S}$  we have
\begin{align}
\Pr(X_{n+1}\in\mathcal{C}')=\Pr(X_{n}\in\mathcal{C}')=\Pr(X_{n}\in\mathcal{C}',X_{n+1}\in\mathcal{C}')+\Pr(X_{n}\in\mathcal{C}',X_{n+1}\in\mathcal{C}).
\end{align}
Subtracting $\Pr(X_{n}\in\mathcal{C}',X_{n+1}\in\mathcal{C}')$ from both sides, establishes the claim.
\end{proof}

\begin{proof}\textbf{of 
Lemma~\ref{lem:Qn_stat}} :
The proof is very similar to the derivation of the invariant measure of a continuous-time Markov chain. Let $\{M'_n\}$ be the process defined as follows:
\begin{enumerate}
    \item Draw $M'_0$ according to the stationary distribution of $M_n$.
    \item For $n> 0$, draw $M'_{n+1} | M'_n \sim W$, where $W$ is the Markov kernel of our chain.
    \item For $n < 0$, draw $M'_{n-1} | M'_n \sim W'$, where $W'$ is the reverse Markov kernel corresponding to the stationary distribution. 
\end{enumerate}
Clearly, $\{M'_n\}$ is a stationary ergodic process with marginal distribution equal to the stationary distribution of $\{M_n\}$. 
Let $Q'_n = \phi(M'_n)$ where $\phi$ is the mapping to the set label (similar to $Q_n$). Clearly, $\{Q'_n\}$ is a stationary process as well, and $\{Q_n\}$ converges to the marginal distribution of $\{Q'_n\}$. Recall that $\{Q'_n\}$ is composed of runs of consecutive letters of $[K]$, and that the length of each run is independent of all past runs. The run-length random variables do depend on the letter $k\in[K]$ of the run, and we denote by $T_k\sim P_{T_k}$ a generic random variable corresponding to a run of the letter $k$. Furthermore, we denote by $A_k(t)$ the event that a new run in $\{Q'_n\}$ of letters $k$ started at time $t$, and let the integer random variable $Z_t\geq 1$ denote the number of symbols left in the current run at time $t$ (including the one at time $t$). 
If $Q'_0=k$, this means that a run of letters $k$ started at some time $-t$, and its corresponding $Z_{-t}$ was greater than $t$. We can therefore write 
\begin{align}
    \pi_k &= \Pr\left(Q'_0=k\right)\\
    &=\sum_{t=0}^\infty \Pr\left(A_k(-t),Z_{-t} > t\right)\\
    &=\sum_{t=0}^\infty \Pr\left(A_k(-t)\right)\Pr\left(T_k> t\right)\label{eq:stat1}\\
    &=\Pr\left(A_k(0)\right) \sum_{t=0}^\infty \Pr\left(T_k> t\right)\label{eq:stat2}\\
    &=\Pr\left(A_k(0)\right) \E\left(T_k\right),\label{eq:expect}
\end{align}
where~\eqref{eq:stat1} follows since given that $A_k(t)$ occurred, $Z_t$ is independent of everything that happened before this run began and has the distribution $P_{T_k}$,~\eqref{eq:stat2} is from stationarity, and~\eqref{eq:expect} is due to the identity $\sum_{t=0}^\infty \Pr\left(T_k> t\right)=\E\left(T_k\right)$ for a non-negative random variable. Thus, from stationarity, for each $t$ we have
\begin{align}
  \Pr\left(A_k(t)\right)=\Pr\left(A_k(0)\right)=\frac{\pi_k}{\E\left(T_k\right)}.\label{eq:ent_prob}  
\end{align}
Now, denote by $B_k(t)$ the event that a run of letters $k$ ended at time $t$. Note that since $\{Q'_n\}$ is stationary, Lemma~\ref{lem:flow} suggests that the probability it enters a state $k$ is equal to the probability it leaves a state $k$ at any given time, namely 
\begin{align}
\Pr(B_k(t))=\Pr(A_k(t))=\frac{\pi_k}{\E\left(T_k\right)},~~~k\in[k].
\end{align}
Now consider the sampled Markov chain $\{S_n\}$, and denote its stationary distribution for state $j$ by $\mu_j$, and its transition probability from state $j$ to state $k$ by $P_{jk}$.
We have 
\begin{align}
\Pr\left(A_k(t+1)\right)& =\sum_{j\neq k}\Pr\left(B_j(t)\right)P_{jk}.\label{eq:stat_eq}
\end{align}
Substituting~\eqref{eq:ent_prob} into~\eqref{eq:stat_eq}, we have
\begin{align}
  \frac{\pi_k}{\E\left(T_k\right)}=\sum_{j\neq k}\frac{\pi_j}{\E\left(T_j\right)} P_{jk} .\label{eq:stat_wait}
\end{align}
Thus, the stationary distribution $\{\pi_k\}_{k\in[K]}$ of $\{Q_n\}$ must satisfy~\eqref{eq:stat_wait}. Since $\{\mu_k\}_{k\in[K]}$ is the unique stationary distribution of $\{S_n\}$, we have that
\begin{align}
      \pi _j^* =\frac{\E[T_j] \mu_j }{\sum_{k=1}^M\E[T_k] \mu _k}
      ,~~j\in[K],
\end{align}
is the unique distribution satisfying~\eqref{eq:stat_wait}, and is consequently the stationary distribution of $\{Q_n\}$, as claimed.
\end{proof}

\begin{proof}\textbf{of Lemma~\ref{lem:trans_prob}} :
By construction, $\{S_n\}$ follows the transition probability law plotted in Figure~\ref{fig:sample_chain}. 
For all $i\in\{2,\ldots,K-2\}$, we have from Lemma~\ref{lem:flow} that by choosing the partition $\mathcal{C}=\{1,\ldots,i-1\},\mathcal{C}'=\{i,\ldots, K\}$ and noting from Figure~\ref{fig:sample_chain} that only adjacent states are connected, $\mu _{i-1}p_{i-1}=\mu _i q_i$, or equivalently
\begin{align}
    \mu _{i-1}&= \frac{q_i}{p_{i-1}}\mu _i\label{eq:stat}.
\end{align}
By construction of the mini-chains $\m{S}_i$ and by Lemma~\ref{lem:N_eps}, we have that $q_i< \epsilon$ and $p_i> 1-\epsilon$ for $i<k$. Thus, repeated application of~\eqref{eq:stat} yields 
\begin{align}
    \mu _{k-i} = \prod_{j=1}^i\frac{q_{k-j+1}}{p_{k-j}}\mu _k \leq  \left(\frac{\epsilon}{1-\epsilon}\right)^{i-1} \mu _{k-1},
\end{align}
for $2\leq i \leq k-1$. Similarly, since $p_i<\epsilon$ and $q_i>1-\epsilon$ for $i>k$, we have
\begin{align}\label{eq:stat_dist}
    \mu _{k+i} = \prod_{j=1}^i\frac{p_{k+j-1}}{q_{k+j}}\mu _k \leq  \left(\frac{\epsilon}{1-\epsilon}\right)^{i-1} \mu _{k+1},
\end{align}
for $1\leq i \leq K-1-k$. For the extreme states $1$ and $K$, by appealing to Lemma~\ref{lem:flow} and recalling that $p_1=1$ and $q_K=1$, we have
\begin{align}
    \mu _1=q_2\mu _2  \leq \epsilon \mu _2 < \frac{\epsilon}{1-\epsilon}\cdot \mu _2,
\end{align}
and
\begin{align}
    \mu _K=p_{K-1} \mu _{K-1} \leq \epsilon \mu _{K-1} < \frac{\epsilon}{1-\epsilon}\cdot \mu _{K-1}.
\end{align}
\end{proof}

\begin{proof}\textbf{of Lemma~\ref{lem:wait_time}} :
Fix $\theta$, and recall that each state $i$ in the sampled chain corresponds to a $\isit\left(N_i,\frac{i}{K+2},\frac{i+1}{K+2}\right)$ mini-chain in the original chain, where $N_i=N(\eps,\frac{i}{K+2},K+2)$ is as defined in~\eqref{eq:NepK}. Restricting our attention to that mini-chain, denote by $s_i=f\left(N_i,\frac{i}{K+2},\frac{i+1}{K+2}\right)$ its initial state, and denote by $T_i^1$ the first time a run of $N_i-s_i$ consecutive ones is observed, and $T_i^0$ as the first time a run of $s_i-1$ consecutive zeros is observed. We exit the mini-chain when either a run of $N_i-s_i$ consecutive ones or a run of $s_i-1$ consecutive zeros is observed, so we have that the exit time $T_i$ satisfies $T_i \leq T_i^1$ and $T_i \leq T_i^0$, which implies
\begin{align}
 \E [T_i] \leq \E [T_i^1],\label{eq:Ti1UB}\\
 \E [T_i] \leq \E [T_i^0].\label{eq:Ti0UB}
\end{align}
Next, we observe that $i\mapsto s_i$ is monotonically non-increasing and $i\mapsto N_i-s_i$ is monotonically non-decreasing. These facts can be verified from the formulas~\eqref{eq:NepK} and~\eqref{eq:start_state} for $N(\eps,\frac{i}{K+2},K+2)$ and  $f\left(N_i,\frac{i}{K+2},\frac{i+1}{K+2}\right)$, respectively. 
Thus the expected time to observe a run of $N_i-s_i$ consecutive ones is also non-decreasing and we have 
\begin{align}
    \E \left[T_1^1\right]\leq \E \left[T_2^1\right]\leq\ldots\leq \E \left[T_j^1\right],\label{eq:onebound}
\end{align}
and similarly
\begin{align}
    \E \left[T_S^0\right]\leq\E \left[T_{S-1}^0\right]\leq\ldots\leq\E \left[T_j^0\right].\label{eq:zerobound}
\end{align}
Let $\{W_n^{j}(\theta)\}$ be a random walk in $\isit\left(N_j,\frac{j}{K},\frac{j+1}{K}\right)$ under $\theta$, and let $W_n^{j}(\theta) \rightarrow 1$ (resp. $W_n^{j}(\theta) \rightarrow 0$)  denote the event that $\{W_n^{j}(\theta)\}$ exits right (resp. exits left). We have
\begin{align}
   T_j^1=T_j+\left(T_j^1-T_j\right)\ind (W_n^{j}(\theta) \rightarrow 0).
\end{align} 
By taking the expectation of both sides, we have
\begin{align}
    \E \left[T_j^1\right]&= \E \left[T_j\right]+\E \left[\left(T_j^1-T_j\right)\ind (W_n^{j}(\theta) \rightarrow 0)\right]\\&=\E \left[T_j\right]+\Pr (W_n^{j}(\theta) \rightarrow 0)\E \left[T_j^1-T_j| W_n^{j}(\theta) \rightarrow 0\right]\\&=\E \left[T_j\right]+\Pr (W_n^{j}(\theta) \rightarrow 0)\E \left[T_j^1\right],
\end{align}
due to
\begin{align}
    \E \left[T_j^1-T_j| W_n^{j}(\theta) \rightarrow 0\right]&=\sum _{t=1}^\infty \Pr (T _j = t | W_n^{j}(\theta)\rightarrow 0) \E \left[T_j^1-T_j|T _j = t, W_n^{j}(\theta) \rightarrow 0\right] \\&=\sum _{t=1}^\infty \Pr (T _j = t | W_n^{j}(\theta) \rightarrow 0) \E \left[T_j^1-t|T_j^1 > t, W_t^{j}(\theta) = 1\right]\label{eq:run_after_run}\\& =\sum _{t=1}^\infty\Pr (T _j = t | W_n^{j}(\theta) \rightarrow 0) \E \left[T_j^1\right]\label{eq:memoryless}\\& =\E \left[T_j^1\right],
\end{align}
where~\eqref{eq:run_after_run} is since no run of $N_j-s_j$ ones was observed until time $t$ and the last bit was $X_t=0$, and~\eqref{eq:memoryless} follows from the memoryless property of the chain. Thus,
\begin{align}
    \E \left[T_j\right]&=\Pr (W_n^{j}(\theta) \rightarrow 1)\E \left[T_j^1\right],\\
    \E \left[T_j\right]&=\Pr (W_n^{j}(\theta)\rightarrow 0)\E \left[T_j^0\right].
\end{align}
Equation~\eqref{eq:uppbound1} now follows by recalling that $\Pr (W_n^{j}(\theta)\rightarrow 0)\geq 1-\epsilon$  for $\theta<\frac{j}{K}$ and by appealing to~\eqref{eq:zerobound} and~\eqref{eq:Ti0UB}.~\eqref{eq:uppbound2} is proven similarly by appealing to~\eqref{eq:onebound}.
\end{proof}

\acks{This work was supported by the ISF under Grants 1791/17 and 1495/18.}

\bibliography{estimation_colt}

\appendix

\section{Bound on the error probability}
To establish Lemma~\ref{lem:N_eps}, we need to prove two supporting lemmas. First, we show that $p^1_0$ is achieved by $\theta=q$, and $p^0_1$ is achieved by $\theta=p$. Denote the probability of deciding $\mathcal{H}_0$ under $\theta$ as $p_0(\theta)$, and the probability of deciding $\mathcal{H}_1$ under $\theta$ as $p_1(\theta)$. 
\begin{lemma}\label{lem:monotone}
For $\isit(N,p,q)$, if $\theta > p$, then $p_1(\theta)\leq p_1(p)$. Similarly, if $\theta<q$, then $p_0(\theta)\leq p_0(q)$.
\end{lemma}
\begin{proof}
We prove the first part of the claim, and the second follows symmetrically. To that end, we use a coupling argument. Denote by $\{W_n^p\}$ the random walk on $\isit(N,p,q)$ under $p$ and $\{W_n^{\theta}\}$ as the random walk on $\isit(N,p,q)$ under $\theta$, where here we assume the extreme states $1$ and $N$ are absorbing, such that once the random walk reaches one of these states, it stays there forever. We couple the two processes using the following joint distribution for $\left(\{W_n^p\},\{W_n^{\theta}\}\right)$: Let $\{W_n^p\}$ be the standard walk on the chain under the $\Bern(p)$ sequence. For any $n$, if $W_n^p$ goes one step to the right, $W_n^{\theta}$ goes one step to the right as well. If $W_n^p$ goes one step to the left, we flip an independent $\Bern\left(\frac{\theta-p}{1-p}\right)$ coin, and $W_n^{\theta}$ goes one step to the right upon seeing $1$ or one step to the left upon seeing $0$. It is easy to see that the marginal distribution under $\{W_n^{\theta}\}$ corresponds to the chain under the $\Bern(\theta)$ distribution, and this is therefore a valid coupling. Our claim now immediately follows from the observation that under this coupling, $W_n^{\theta}$ is never to the left of $W_n^p$. 
\end{proof}
Second, we prove the following lemma, which bounds the error probability of $\isit(N,p,q)$ when the hypotheses are $\frac{1}{K}$ apart. 
\begin{lemma}\label{lem:close_hyp}
For any $\frac{2}{K}\leq p\leq 1-\frac{1}{K}$, $q=p-\frac{1}{K}$ and $N\geq 3+\left\lceil K\cdot 6\log\frac{2}{ p_{\min}} \right\rceil$, it holds that
\begin{align}
    \Pe^{\isit(N,p,q)}\leq  \frac{2}{p_{\min}}\cdot \exp_2\left\{-\frac{ \left(1-\frac{1}{K\cdot p}\right)\left(\frac{1}{K}H_b(p)-\frac{1}{K^2}\log p\right)\left(N-3\right)}{\frac{1}{K} \left(1-2\left(p-\frac{1}{K}\right)\right)-2\left(p-\frac{1}{K}\right)\left(1-p+\frac{1}{K}\right)\log p(1-p)}\right\},
\end{align}
where $p_{\min}=\min \{p(1-p),q(1-q)\}$ and $H_b(p)\triangleq -p\log p -(1-p)\log (1-p)$ is the binary entropy of $p$.  
\end{lemma}
\begin{proof} 
In~\cite{berg2020binary}, the authors showed that for initial state $s$, we have 
\begin{align}
p_1(p)&=\frac{1-p^{N-s}}{1+\frac{p^{N-s-1}}{(1-p)^{s-2}}-p^{N-s-1}}\\
&\leq \frac{(1-p)^{s-2}}{p^{N-s-1}}\cdot \frac{1}{1-(1-p)^{s-2}},
\end{align}
and
\begin{align}
  p_0(q)&= \frac{1-(1-q)^{s-1}}{1+\frac{(1-q)^{s-2}}{q^{N-s-1}}-(1-q)^{s-2}}\\
  &\leq \frac{q^{N-s-1}}{(1-q)^{s-2}}\cdot \frac{1}{1-q^{N-s-1}}.
\end{align}
Choosing $s=s^*$, where $s^*$ is  
\begin{align}
    2+\frac{\log pq}{\log p(1-p)+\log q(1-q)}(N-3),
    \label{eq:optimal_s}
\end{align}
we get
\begin{align}
   \frac{(1-p)^{s^*-2}}{p^{N-s^*-1}} =\frac{q^{N-s^*-1}}{(1-q)^{s^*-2}}
   =2^{-r(p,q)(N-3)},\label{eq:BHT_exp}
\end{align}
where
\begin{align}
    r(p,q)\triangleq \frac{\log p \log (1-q)-\log q \log (1-p)}{\log p(1-p)+\log q(1-q)}.
\end{align}
We therefore have, for $s=s^*$,
\begin{align}
    \max \left\{ p_0(q),p_1(p)\right\}\leq 2^{-r(p,q)(N-3)}\cdot \max \left\{ \frac{1}{1-(1-p)^{s^*-2}},\frac{1}{1-q^{N-s^*-1}}\right\}\label{eq:error}.
\end{align}
Recall that $s$ is a state in the chain so it must be an integer, whereas $s^*$ may not be. Thus, we need to round $s^*$ either up or down, in which case, both ratios in~\eqref{eq:BHT_exp} $\frac{(1-p)^{s^*-2}}{p^{N-s^*-1}}$, and $\frac{q^{N-s^*-1}}{(1-q)^{s^*-2}}$, will increase by at most $\frac{1}{p_{\min}}$, where $p_{\min}=\min \{p(1-p),q(1-q)\}$. Furthermore, for our choice of $N$, $\frac{2}{K}\leq p\leq 1-\frac{1}{K}$ and $q=p-\frac{1}{K}$, we have that $3<s*< N-2$ and the rightmost part of~\eqref{eq:error} is always upper bounded by $2$. Combining this with Lemma~\ref{lem:monotone}, we therefore get the bound
\begin{align}
   \Pe^{\isit(N,p,q)}=\max\left\{p^0_1,p^1_0\right\}=\max \left\{ p_0(q),p_1(p)\right\} \leq  \frac{2}{p_{\min}} \cdot 2^{-r(p,q)(N-3)}.\label{eq:pe_run}
\end{align}
Setting $p-q=\delta>0$, we have
\begin{align}
    r\left(p,p-\delta\right) &=\frac{\log p \log (1-p+\delta)-\log (p-\delta) \log (1-p)}{\log p(1-p)+\log (p-\delta)(1-p+\delta)}\\&=\frac{\log p \left(\log (1-p)+\log \left(1+\frac{\delta}{1-p}\right)\right)-\left(\log p+\log \left(1-\frac{\delta}{p}\right) \right)\log (1-p)}{\log p(1-p)+\log (1-p)+\log \left(1+\frac{\delta}{1-p}\right)+\log p+\log \left(1-\frac{\delta}{p}\right) }
     \\&\geq
    \frac{\frac{\delta}{1-p+\delta}\log p+\frac{\delta}{p}\log (1-p)}{2\log p(1-p)+\frac{\epsilon}{1-p+\delta}-\frac{\delta}{p-\delta}}\label{eq:log_ineq}
    \\&=
     -\frac{p-\delta}{p}\cdot \frac{\delta p\log p +\delta(1-p+\delta)\log (1-p)}{2(p-\delta)(1-p+\delta)\log p(1-p)-\delta (1-2(p-\delta))}
    \\& =\left(1-\frac{\delta}{p}\right) \cdot \frac{\delta H_b(p)-\delta^2\log (1-p)}{\delta (1-2(p-\delta))-2(p-\delta)(1-p+\delta)\log p(1-p)}\label{eq:bin_ent},
\end{align}
where~\eqref{eq:log_ineq} follows from $\frac{x}{x+1}\leq \log (1+x) \leq x$ and~\eqref{eq:bin_ent} follows from the definition of the binary entropy. The claim follows by substituting $\delta = \frac{1}{K}$.
\end{proof}
\begin{proof}\textbf{of Lemma~
\ref{lem:N_eps}} :
Let $N=3+\lceil c\cdot K\rceil $, for some $c\geq 6\log\frac{2}{ p_{\min}}$. From Lemma~\ref{lem:close_hyp},
\begin{align}
    \Pe^{\isit\left(N,p,p-\frac{1}{K}\right)} \leq  \frac{2}{p_{\min}}\cdot \exp_2\left\{-\frac{ c \left(1-\frac{1}{K\cdot p}\right)\left(H_b(p)-\frac{1}{K}\log (1-p)\right)}{\frac{1}{K} \left(1-2\left(p-\frac{1}{K}\right)\right)-2\left(p-\frac{1}{K}\right)\left(1-p+\frac{1}{K}\right)\log p(1-p)}\right\}
\end{align}
In order to guarantee $\Pe^{\isit\left(N,p,p-\frac{1}{K}\right)} \leq \epsilon$, it is sufficient to choose $c$ to be
\begin{align}
    \frac{\frac{1}{K} \left(1-2\left(p-\frac{1}{K}\right)\right)-2\left(p-\frac{1}{K}\right)\left(1-p+\frac{1}{K}\right)\log p(1-p)}{\left(1-\frac{1}{K\cdot p}\right)\left(H_b(p)-\frac{1}{K}\log (1-p)\right)}\cdot \log \frac{2}{\epsilon p_{\min}}.\label{eq:c_exp}
\end{align}
Upper bounding the first term in the brackets, we get
\begin{align}
&\frac{\frac{1}{K} \left(1-2\left(p-\frac{1}{K}\right)\right)-2\left(p-\frac{1}{K}\right)\left(1-p+\frac{1}{K}\right)\log p(1-p)}{\left(1-\frac{1}{K\cdot p}\right)\left(H_b(p)-\frac{1}{K}\log (1-p)\right)}\\
&\leq \frac{1}{1-\frac{1}{K\cdot p}}\cdot \frac{\frac{1}{K}+2\left(H_b(p)-\frac{1}{K}\log (1-p)\right)}{H_b(p)-\frac{1}{K}\log (1-p)}\label{eq:bin_ent_bound}\\&= \frac{1}{1-\frac{1}{K\cdot p}}
\left(2+\frac{1}{K\cdot H_b(p)-\log(1-p)}\right)
\\&\leq  \frac{3}{1-\frac{1}{K\cdot p}}\label{eq:weird_bound} \\&\leq 6\label{eq:p_bound},
\end{align}
where~\eqref{eq:bin_ent_bound},~\eqref{eq:weird_bound} and~\eqref{eq:p_bound} follows since $p\geq \frac{2}{K}$ implies 
\begin{enumerate}[(i)]
    \item $H_b(p)-\frac{1}{K}\log (1-p)\geq -\left(p-\frac{1}{K}\right)\left(1-p+\frac{1}{K}\right)\log p(1-p)$,
    \item $K\cdot H_b(p)-\log(1-p)\geq 1$,
    \item $\frac{1}{1-\frac{1}{K\cdot p}}\leq 2$.
\end{enumerate}
Combining~\eqref{eq:p_bound} and~\eqref{eq:c_exp}, noting that $\min \left\{p(1-p),\left(p-\frac{1}{K}\right)\left(1-p+\frac{1}{K}\right)\right\}\geq \left(p-\frac{1}{K}\right)(1-p)$, and choosing
\begin{align}
    c=c_{\epsilon,p} = 6\log \frac{2}{\epsilon \left(p-\frac{1}{K}\right)(1-p)},
\end{align}
the proof is concluded.
\end{proof} 

\section{Calculation of number of states $S$ in~\eqref{eq:Ssum}}
Using the expression in~\eqref{eq:NepK} for $N(\eps,p,K)$ we obtain
\begin{align}
S&=\sum_{k=1}^K N_k\\&=\sum_{k=1}^K N\left(\eps,\frac{k+1}{K+2},K+2\right)\\&\leq 4K+6(K+2)\sum_{k=1}^K \log  \frac{2}{\eps\left(\frac{k}{K+2}\cdot\frac{K-k+1}{K+2}\right)}  \\
&= 4K+6K(K+2)\log \left(\frac{2}{\epsilon}\right) -6(K+2)\cdot 2 \sum_{k=1}^\frac{K}{2}\log \left(\frac{k}{K+2}\cdot\frac{K-k+1}{K+2}\right)\label{eq:symmetry}\\& \leq 4K+6K(K+2)\log \left(\frac{2}{\epsilon}\right)-6(K+2)\cdot 2 \sum_{k=1}^\frac{K}{2}\log \left(\frac{k}{K+2}\right)-6K(K+2)\label{eq:log_bound1}\\& \leq 4K+6K(K+2)\log \left(\frac{2}{\epsilon}\right)-6K(K+2) \log \left(\frac{K}{2e(K+2)}\right)-6K(K+2)\label{eq:factorial_bound}\\ & \leq 4K+6K(K+2)\log \left(\frac{2e}{\epsilon}\right)+12(K+2)\leq 6(K+2)^2\log \left(\frac{2e}{\epsilon}\right)\label{eq:log_bound2b},
\end{align}
where~\eqref{eq:symmetry} follows from the symmetry of $\left(\frac{k}{K+2}\cdot\frac{K-k+1}{K+2}\right)$ around $k=\frac{K}{2}$,~\eqref{eq:log_bound1} from $\frac{K-k+1}{K+2}\geq \frac{1}{2}$ for all $1\leq k \leq \frac{K}{2}$,~\eqref{eq:factorial_bound} is from $n!\geq \left(\frac{n}{e}\right)^n$ and~\eqref{eq:log_bound2b} follows from $\log (1+x) \geq \frac{x}{x+1}$. 

\section{Proof of $R_{\theta}=O(1/S)$ for $\theta\in \left[0,\frac{1}{K+2}\right)$ and $\theta\in \left(\frac{K+1}{K+2},1\right]$}
We shall prove the case $\theta \leq \frac{1}{K+2}$. The case of $\theta \geq 1- \frac{1}{K+2}$ follows from a symmetric argument.
We show how previous results imply that for very small $\theta$ the stationary distribution is concentrated on the two leftmost states of the sampled chain. From there, the proof is similar (yet not identical) to the proof of the general case. Let us go step by step:

\begin{itemize}
    \item Firstly, Lemma~\ref{lem:N_eps} implies that $p_k < \epsilon$ for all $k>1$ in the chain of Figure~\ref{fig:sample_chain}.
    \item Now, a simplified (one-sided) version of Lemma~\ref{lem:trans_prob} shows the stationary distribution is exponentially decreasing for all states $\geq 2$. This follows since eq.~\eqref{eq:stat_dist} still holds with $k=1$,
    \begin{align}
    \mu _{i+1} \leq  \left(\frac{\epsilon}{1-\epsilon}\right)^{i-1} \mu _{2},
    \end{align}
    for $1\leq i \leq K-1$.
    \item Applying Lemma~\ref{lem:wait_time}, eq.~\eqref{eq:uppbound1} states that $\E[ T_ j ] > (1-\epsilon) \E[ T_ i]$ for all $j \in [n]$ and $i>j$.
    \item Calculate the risk $R_{\theta}$.
\end{itemize}
\begin{align}
   R_\theta &=\sum_{i=1}^K\frac{\E[T_i] \mu_i }{\sum_{j=1}^K\E[T_j] \mu _j}\left(\frac{i}{K+2}-\theta\right)^2
  \\&=\frac{\E[T_1] \mu_1 }{\sum_{j=1}^K\E[T_j] \mu _j}\left(\frac{1}{K+2}-\theta\right)^2+\sum_{i=2}^K\frac{\E[T_i] \mu_i }{\sum_{j=1}^K\E[T_j] \mu _j}\left(\frac{i}{K+2}-\theta\right)^2
  \\ &\leq \frac{1}{(K+2)^2}+\frac{1}{1-\epsilon}\sum_{i=2}^K\frac{\E[T_{2}] \mu_{2} }{\sum_{j=1}^K\E[T_j] \mu _j}\frac{\mu_i}{\mu_{2}}\left(\frac{i}{K+2}-\theta\right)^2\label{eq:exp_bound2}\\&\leq  
 \frac{1}{(K+2)^2}+ \frac{1}{1-\epsilon}\sum_{i=1}^{K-1}\left(\frac{\epsilon}{1-\epsilon}\right)^{i-1} \left(\frac{i+1}{K+2}\right)^2\label{eq:stationary_bound2}\\&\leq \frac{1}{(K+2)^2}\cdot \frac{1}{1-\epsilon}\left(\sum_{i=1}^{\infty} \left(\frac{\epsilon}{1-\epsilon}\right)^{i-1}(i+1)^2+1\right)\label{eq:sum_bound2}\\&\leq \frac{6\log \left(\frac{2e}{\epsilon}\right)}{S}\left(\frac{\epsilon}{(1-2\epsilon)^3}+\frac{4(1-\epsilon)}{(1-2\epsilon)^2}+\frac{1}{1-\epsilon}\right),\label{eq:sum_identity2}
\end{align}
where~\eqref{eq:exp_bound2} follows from Lemma~\ref{lem:wait_time}, ~\eqref{eq:stationary_bound2} follows from Lemma~\ref{lem:trans_prob} and since $\frac{\E[T_j] \mu_j }{\sum_{k=1}^M\E[T_k] \mu _k}\leq 1$,~\eqref{eq:sum_bound2} is since we only add positive terms, and~\eqref{eq:sum_identity2} is due to the identity $\sum_{i=0}^\infty q^i (i+2)^2=\frac{q(1+q)+4(1-q)}{(1-q)^3}$ and by substituting~\eqref{eq:Ssum}. Finally, substituting $\eps=1/100$ into~\eqref{eq:sum_identity2} gives $R_{\theta}<\frac{300}{S}$.

\end{document}